\def\final{0}

\documentclass[11pt]{article}

\usepackage{algorithm, algorithmic}
\usepackage{amsmath, amssymb, amsthm}
\usepackage{framed}
\usepackage{fullpage}
\usepackage[margin=1.1in]{geometry}
\usepackage{verbatim}
\usepackage{color}
\usepackage{kpfonts}
\definecolor{DarkGreen}{rgb}{0.1,0.5,0.1}
\definecolor{DarkRed}{rgb}{0.6,0.2,0.2}
\definecolor{DarkBlue}{rgb}{0.2,0.2,0.6}
\usepackage[pdftex]{hyperref}
\hypersetup{
    unicode=false,				
    pdftoolbar=true,				
    pdfmenubar=true,			
    pdffitwindow=false, 		
    pdftitle={},    					
    pdfauthor={}
    pdfsubject={},				
    pdfnewwindow=true,		
    pdfkeywords={},			
    colorlinks=true,				
    linkcolor=DarkBlue,		
    citecolor=DarkBlue,	
    filecolor=DarkRed,		
    urlcolor=DarkBlue,		
}

\ifnum\final=0
\newcommand{\mynote}[1]{\marginpar{\tiny\sf #1}}
\else
\newcommand{\mynote}[1]{}
\fi
\newcommand{\jnote}[1]{}
\newcommand{\tnote}[1]{}

\newcommand{\INDSTATE}[1][1]{\STATE\hspace{#1\algorithmicindent}}

\newcommand{\pr}[2]{\underset{#1}{\mathbb{P}}\left[ #2 \right]}
\newcommand{\ex}[2]{\underset{#1}{\mathbb{E}}\left[ #2 \right]}
\newcommand{\var}[2]{\underset{#1}{\mathrm{Var}}\left[ #2 \right]}
\newcommand{\norm}[1]{\left|\left| #1 \right| \right|}

\newcommand{\pmo}{\{\pm1\}}

\newcommand{\from}{:}

\newcommand{\eps}{\varepsilon}

\newcommand{\cR}{\mathcal{R}}

\newcommand{\trace}{\mathit{Trace}}

\newtheorem{theorem}{Theorem}[section]
\newtheorem{thm}[theorem]{Theorem}

\newtheorem{lem}[theorem]{Lemma}
\newtheorem{fact}[theorem]{Fact}

\newtheorem{prop}[theorem]{Proposition}

\theoremstyle{definition}

\newtheorem{definition}[theorem]{Definition}

\title{Between Pure and Approximate Differential Privacy}
\author{Thomas Steinke\thanks{Harvard University School of Engineering and Applied Sciences.  Supported by NSF grant CCF-1116616. \newline Email: \href{mailto:tsteinke@seas.harvard.edu}{tsteinke@seas.harvard.edu}.} \and Jonathan Ullman\thanks{Columbia University Department of Computer Science.  Supported by a Junior Fellowship from the Simons Society of Fellows.  Email: \href{mailto:jullman@cs.columbia.edu}{jullman@cs.columbia.edu}.}}

\begin{document}
\maketitle

\pagenumbering{gobble}
\begin{abstract}
We show a new lower bound on the sample complexity of $(\eps, \delta)$-differentially private algorithms that accurately answer statistical queries on high-dimensional databases.  The novelty of our bound is that it depends optimally on the parameter $\delta$, which loosely corresponds to the probability that the algorithm fails to be private, and is the first to smoothly interpolate between approximate differential privacy ($\delta > 0$) and pure differential privacy ($\delta = 0$).

Specifically, we consider a database $D \in \pmo^{n \times d}$ and its \emph{one-way marginals}, which are the $d$ queries of the form ``What fraction of individual records have the $i$-th bit set to $+1$?''  We show that in order to answer all of these queries to within error $\pm \alpha$ (on average) while satisfying $(\eps, \delta)$-differential privacy, it is necessary that
$$
n \geq \Omega\left( \frac{\sqrt{d \log(1/\delta)}}{\alpha \eps} \right),
$$
which is optimal up to constant factors.   To prove our lower bound, we build on the connection between \emph{fingerprinting codes} and lower bounds in differential privacy (Bun, Ullman, and Vadhan, STOC'14).

In addition to our lower bound, we give new purely and approximately differentially private algorithms for answering arbitrary statistical queries that improve on the sample complexity of the standard Laplace and Gaussian mechanisms for achieving worst-case accuracy guarantees by a logarithmic factor.
\end{abstract}

\vfill
\newpage

\tableofcontents

\vfill
\newpage

\pagenumbering{arabic}
\section{Introduction}
The goal of privacy-preserving data analysis is to enable rich statistical analysis of a database while protecting the privacy of individuals whose data is in the database. A formal privacy guarantee is given by \emph{$(\eps, \delta)$-differential privacy}~\cite{DworkMNS06,DworkKMMN06}, which ensures that no individual's data has a significant influence on the information released about the database. The two parameters $\eps$ and $\delta$ control the level of privacy.  Very roughly, $\eps$ is an upper bound on the amount of influence an individual's record has on the information released and $\delta$ is the probability that this bound fails to hold\footnote{This intuition is actually somewhat imprecise, although it is suitable for this informal discussion.  See \cite{KasiviswanathanS08} for a more precise semantic interpretation of $(\eps, \delta)$-differential privacy.}, so the definition becomes more stringent as $\eps, \delta \rightarrow 0$.

A natural way to measure the tradeoff between privacy and utility is \emph{sample complexity}---the minimum number of records $n$ that is sufficient in order to publicly release a given set of statistics about the database, while achieving both differential privacy and statistical accuracy.  Intuitively, it's easier to achieve these two goals when $n$ is large, as each individual's data will have only a small influence on the aggregate statistics of interest.  Conversely, the sample complexity $n$ should increase as $\eps$ and $\delta$ decrease (which strengthens the privacy guarantee).  

The strongest version of differential privacy, in which $\delta = 0$, is known as \emph{pure differential privacy}.   The sample complexity of achieving pure differential privacy is well known for many settings (e.g.~\cite{HardtT10}).  The more general case where $\delta > 0$ is known as \emph{approximate differential privacy}, and is less well understood.  Recently, Bun, Ullman, and Vadhan~\cite{BunUV14} showed how to prove strong lower bounds for approximate differential privacy that are essentially optimal for $\delta \approx 1/n$, which is essentially the weakest privacy guarantee that is still meaningful.\footnote{When $\delta \geq 1/n$ there are algorithms that are intuitively not private, yet  satisfy $(0, \delta)$-differential privacy.}  

Since $\delta$ bounds the probability of a complete privacy breach, we would like $\delta$ to be very small. Thus we would like to quantify the cost (in terms of sample complexity) as $\delta \rightarrow 0$.  In this work we give lower bounds for approximately differentially private algorithms that are nearly optimal for every choice of $\delta$, and smoothly interpolate between pure and approximate differential privacy.  

Specifically, we consider algorithms that compute the \emph{one-way marginals of the database}---an extremely simple and fundamental family of queries\jnote{Probably want to include some better motivation for looking at this case.  I'm sure it's not obvious to non-experts why one-way marginals are fundamental.}\tnote{Agreed. One good motivatioin would be to talk about implications of the result. e.g. answering $O(n^2)$ queries is hard.}.  For a database $D \in \pmo^{n \times d}$, the $d$ one-way marginals are simply the mean of the bits in each of the $d$ columns.  Formally, we define
$$
\overline{D} := \frac{1}{n} \sum_{i=1}^{n} D_i \in [\pm 1]^d
$$
where $D_i \in \{\pm 1\}^d$ is the $i$-th row of $D$. A mechanism $M$ is said to be \emph{accurate} if, on input $D$, its output is ``close to'' $\overline{D}$. Accuracy may be measured in a \emph{worst-case} sense---i.e. $\norm{M(D)-\overline{D}}_\infty \leq \alpha$, meaning every one-way marginal is answered with accuracy $\alpha$---or in an \emph{average-case} sense---i.e. $\norm{M(D)-\overline{D}}_1 \leq \alpha d$, meaning the marginals are answered with average accuracy $\alpha$.

Some of the earliest results in differential privacy \cite{DinurN03, DworkN04, BlumDMN05, DworkMNS06} give a simple $(\eps, \delta)$-differentially private algorithm---the \emph{Laplace mechanism}---that computes the one-way marginals of $D \in \pmo^{n \times d}$ with average error $\alpha$ as long as \begin{equation} \label{eqn:LaplaceMechanism} n \geq O\left(\min \left\{\frac{\sqrt{d \log(1/\delta)}}{\eps \alpha}, \frac{d}{\varepsilon \alpha} \right\}\right).\end{equation}  The previous best lower bounds are $n \geq \Omega(d/\varepsilon\alpha)$ \cite{HardtT10} for pure differential privacy and $n \geq \tilde\Omega(\sqrt{d}/ \eps \alpha) $ for approximate differential privacy with $\delta = o(1/n)$ \cite{BunUV14}. \tnote{I'm not familiar with this $n \gtrsim  \log(1/\delta) / \eps \alpha$ bound.} Our main result is an optimal lower bound that combines the previous lower bounds.

\jnote{I like the statement of the theorem with $\log_d$ in place of the $\tilde{\Omega}$.}
\begin{thm}[Main Theorem] \label{thm:maininformal}
For every $\eps \leq O(1)$, every $2^{-\Omega(n)} \leq \delta \leq 1/n^{1+\Omega(1)}$ and every \linebreak $\alpha \leq 1/10$, if $M \from \pmo^{n \times d} \to [\pm 1]^d$ is $(\eps, \delta)$-differentially private and $\ex{M}{\| M(D) - \overline{D} \|_1} \leq \alpha d$, then
$$
n \geq {\Omega}\left( \frac{\sqrt{d \log(1/\delta)}}{\eps \alpha}\right).
$$
\end{thm}

More generally, this is the first result showing that the sample complexity must grow by a multiplicative factor of $\sqrt{\log(1/\delta)}$ for answering any family of queries, as opposed to an additive dependence on $\delta$.  We also remark that the assumption on the range of $\delta$ is necessary, as the Laplace mechanism gives accuracy $\alpha$ and satisfies $(\eps, 0)$-differential privacy when $n \geq O(d/\eps \alpha)$.


\subsection{Average-Case Versus Worst-Case Error}

Our lower bound holds for mechanisms with an average-case ($L_1$) error guarantee.  Thus, it also holds for algorithms that achieve worst-case ($L_\infty$) error guarantees.  The Laplace mechanism gives a matching upper bound for average-case error. In many cases worst-case error guarantees are preferrable. For worst-case error, the sample complexity of the Laplace mechanism degrades by an additional $\log d$ factor compared to \eqref{eqn:LaplaceMechanism}.

Surprisingly, this degradation is not necessary.\tnote{Not sure if "surprisingly" is the appropriate word, but I think it is somewhat surprising that indepenent noise is suboptimal in this setting.}
We present algorithms that answer every one-way marginal with $ \alpha$ accuracy and improve on the sample complexity of the Laplace mechanism by roughly a $\log d$ factor.\jnote{Should probably say something more exciting.}  These algorithms demonstrate that the widely used technique of adding independent noise to each query is suboptimal when the goal is to achieve worst-case error guarantees.

Our algorithm for pure differential privacy satisfies the following.
\begin{thm} \label{thm-intro:LinfPure}
For every $\varepsilon, \alpha > 0$,  $d \geq 1$,  and $n \geq 4d/\varepsilon\alpha $, there exists an efficient mechanism $M : \pmo^{n \times d} \to [\pm 1]^d$  that is $(\varepsilon, 0)$-differentially private and $$\forall D \in \pmo^{n \times d} ~~~~ \pr{M}{\norm{M(D)-\overline{D}}_\infty \geq \alpha} \leq (2e)^{-d}.$$
\end{thm}

And our algorithm for approximate differential privacy is as follows.
\begin{thm} \label{thm-intro:LinfApprox}
For every $\varepsilon, \delta, \alpha > 0$, $d \geq 1$, and 
$$
n \geq O\left(\frac{\sqrt{d \cdot \log(1/\delta) \cdot \log \log d}}{\eps \alpha}\right),
$$ there exists an efficient mechanism $M : \pmo^{n \times d} \to [\pm 1]^d$  that is $(\varepsilon, \delta)$-differentially private and $$\forall D \in \pmo^{n \times d} ~~~~ \pr{M}{\norm{M(D)-\overline{D}}_\infty \geq \alpha} \leq \frac{1}{d^{\omega(1)}}.$$
\end{thm}

These algorithms improve over the sample complexity of the best known mechanisms for each privacy and accuracy guarantee by a factor of $(\log(d))^{\Omega(1)}$.  Namely, the Laplace mechanism requires $n \geq O(d \cdot \log d / \eps \alpha)$ samples for pure differential privacy and the Gaussian mechanism requires $n \geq O(\sqrt{d \cdot  \log(1/\delta) \cdot \log d} / \eps \alpha)$ samples for approximate differential privacy.

\begin{figure}[h]
\begin{center}
\begin{tabular}{|c|c|c||cc|c|}
\hline
Privacy&Accuracy&Type&\multicolumn{2}{c|}{Previous bound}&This work\\
\hline
\hline
$(\varepsilon,\delta)$&$L_1$ or $L_\infty$&Lower&$n = \tilde{\Omega}\left(\frac{\sqrt{d}}{\alpha\varepsilon}\right)$&\cite{BunUV14}&$n=\Omega\left(\frac{\sqrt{d \log(1/\delta)}}{\alpha\varepsilon}\right)$\\
\hline
$(\varepsilon,\delta)$&$L_1$&Upper&$n=O\left(\frac{\sqrt{d \cdot \log(1/\delta)}}{\varepsilon\alpha}\right)$&Laplace&\\
\hline
$(\varepsilon,\delta)$&$L_\infty$&Upper&$n=O\left(\frac{\sqrt{d \cdot \log(1/\delta) \cdot \log d}}{\varepsilon\alpha}\right)$&Gaussian&$n=O\left(\frac{\sqrt{d \cdot \log(1/\delta) \cdot \log \log d}}{\varepsilon\alpha}\right)$\\
\hline
\hline
$(\varepsilon,0)$&$L_1$ or $L_\infty$&Lower&$n = \Omega\left(\frac{d}{\alpha\varepsilon}\right)$&\cite{HardtT10}&\\
\hline
$(\varepsilon,0)$&$L_1$&Upper&$n=O\left(\frac{d}{\varepsilon\alpha}\right)$&Laplace&\\
\hline
$(\varepsilon,0)$&$L_\infty$&Upper&$n=O\left(\frac{d \cdot \log d}{\varepsilon\alpha}\right)$&Laplace&$n=O\left(\frac{d}{\varepsilon\alpha}\right)$\\
\hline
\end{tabular}
\caption{Summary of sample complexity upper and lower bounds for privately answering $d$ one-way marginals with accuracy $\alpha$.}
\end{center}
\end{figure}

\subsection{Techniques} \tnote{Techniques section needs work}

\paragraph{Lower Bounds:}
Our lower bound relies on a combinatorial objected called a \emph{fingerprinting code} \cite{BonehS98}. Fingerprinting codes were originally used in cryptography for watermarking digital content, but several recent works have shown they are intimately connected to lower bounds for differential privacy and related learning problems \cite{Ullman13, BunUV14, HardtU14, SteinkeU14}.  In particular, Bun et al.~\cite{BunUV14} showed that fingerprinting codes can be used to construct an attack demonstrating that any mechanism that accurately answers one-way marginals is not differentially private. Specifically, a fingerprinting code gives a distribution on individuals' data and a corresponding ``tracer'' algorithm such that, if a database is constructed from the data of a fixed subset of the individuals, then the tracer algorithm can identify at least one of the individuals in that subset given only approximate answers to the one-way marginals of the database.  Specifically, their attack shows that a mechanism that satisfies $(1, o(1/n))$-differential privacy requires $n \geq \tilde{\Omega}(\sqrt{d})$ samples to accurately compute one-way marginals.

Our proof uses a new, more general reduction from breaking fingerprinting codes to differentially private data release.  Specifically, our reduction uses \emph{group differential privacy}.  This property states that if an algorithm is $(\eps, \delta)$-differentially private with respect to the change of one individual's data, then for any $k$, it is roughly $(k \eps, e^{k\eps} \delta)$-differentially private with respect to the change of $k$ individuals' data.  Thus an $(\eps, \delta)$-differentially private algorithm provides a meaningful privacy guarantee for groups of size $k \approx \log(1/\delta)/\eps$.

To use this in our reduction, we start with a mechanism $M$ that takes a database of $n$ rows and is $(\eps, \delta)$-differentially private.  We design a mechanism $M_k$ that takes a database of $n/k$ rows, copies each of its rows $k$ times, and uses the result as input to $M$. The resulting mechanism $M_k$ is roughly $(k \eps, e^{k \eps} \delta)$-differentially private.  For our choice of $k$, these parameters will be small enough to apply the attack of \cite{BunUV14} to obtain a lower bound on the number of samples used by $M_k$, which is $n/k$.  Thus, for larger values of $k$ (equivalently, smaller values of $\delta$), we obtain a stronger lower bound.  The remainder of the proof is to quantify the parameters precisely.
\jnote{Do we want to include something about the upper bounds?}

\paragraph{Upper Bounds:} Our algorithm for pure differential privacy and worst-case error is an instantiation of the exponential mechanism \cite{McSherryT07} using the $L_\infty$ norm. That is, the mechanism samples $y \in \mathbb{R}^d$ with probability proportional to $\exp(-\eta \norm{y}_\infty)$ and outputs $M(D)=\overline{D}+y$.  In contrast, adding independent Laplace noise  corresponds to using the exponential mechanism with the $L_1$ norm and adding independent Gaussian noise corresponds to using the exponential mechanism with the $L_2$ norm squared. Using this distribution turns out to give better tail bounds than adding independent noise.

For approximate differential privacy, we use a completely different algorithm. We start by adding independent Gaussian noise to each marginal. However, rather than using a union bound to show that each Gaussian error is small with high probability, we use a Chernoff bound to show that most errors are small. Namely, with the sample complexity that we allow $M$, we can ensure that all but a $1/\mathrm{polylog}(d)$ fraction of the errors are small. Now we ``fix'' the $d/\mathrm{polylog}(d)$ marginals that are bad. The trick is that we use the sparse vector algorithm, which allows us to do indentify and fix these $d/\mathrm{polylog}(d)$ marginals with sample complexity corresponding to only $d/\mathrm{polylog}(d)$ queries, rather than $d$ queries.

\section{Preliminaries}
We define a \emph{database} $D \in \pmo^{n \times d}$ to be a matrix of $n$ rows, where each row corresponds to an individual, and each row has \emph{dimension $d$} (consists of $d$ binary attributes).  We say that two databases $D, D' \in \pmo^{n \times d}$ are \emph{adjacent} if they differ only by a single row, and we denote this by $D \sim D'$.  In particular, we can replace the $i$th row of a database $D$ with some fixed element of $\pmo^d$ to obtain another database $D_{-i} \sim D$.  
\begin{definition}[Differential Privacy~\cite{DworkMNS06}]\label{def:dp} 
Let $M \from \pmo^{n \times d} \to \cR$ be a randomized mechanism. We say that $M$ is \emph{$(\eps, \delta)$-differentially private} if for every two adjacent databases $D \sim D'$ and every subset $S \subseteq \cR$,
$$
\pr{}{M(D) \in S} \leq e^{\eps} \cdot \pr{}{M(D') \in S}+ \delta.
$$
\end{definition}

A well known fact about differential privacy is that it generalizes smoothly to databases that differ on more than a single row.  We say that two databases $D, D' \in \pmo^{n \times d}$ are \emph{$k$-adjacent} if they differ by at most $k$ rows, and we denote this by $D \sim_{k} D'$.
\jnote{I think this is very slightly wrong.  Can you double check?}
\begin{fact}[Group Differential Privacy] \label{fact:GroupDP}
For every $k \geq 1$, if $M \from \pmo^{n \times d} \to \cR$ is $(\eps, \delta)$-differentially private, then for every two $k$-adjacent databases $D \sim_{k} D'$, and every subset $S \subseteq \cR$,
$$
\pr{}{M(D) \in S} \leq e^{k \eps} \cdot \pr{}{M(D') \in S} + \frac{e^{k \eps} - 1}{e^{\eps} - 1} \cdot \delta.
$$
\end{fact}

\jnote{Do you think there is a better way to state this?}

All of the upper and lower bounds for one-way marginals have a multiplicative $1/\alpha\varepsilon$ dependence on the accuracy $\alpha$ and the privacy loss $\varepsilon$. This is no coincidence - there is a generic reduction:

\begin{fact}[$\alpha$ and $\varepsilon$ dependence]\label{fact:alphaeps} Let $p \in [1,\infty]$ and $\alpha,\varepsilon, \delta \in [0,1/10]$. 

Suppose there exists a $(\varepsilon,\delta)$-differentially private mechanism $M \from \pmo^{n \times d} \to [\pm 1]^d$ such that for every database $D \in \pmo^{n \times d}$,
$$
\ex{M}{\|M(D) - \overline{D}\|_p} \leq \alpha d^{1/p}.
$$

Then there exists a $(1,\delta/\varepsilon)$-differentially private mechanism $M' \from \pmo^{n' \times d} \to [\pm 1]^d$ for $n'=\Theta(\alpha\varepsilon n)$ such that for every database $D' \in \pmo^{n' \times d}$,
$$
\ex{M'}{\|M'(D') - \overline{D'}\|_p} \leq d^{1/p}/10.
$$
\end{fact}
This fact allows us to suppress the accuracy parameter $\alpha$ and the privacy loss $\varepsilon$ when proving our lower bounds.
Namely, if we prove a lower bound of $n' \geq n^*$ for all $(1,\delta)$-differentially private mechanisms $M' \from \pmo^{n' \times d} \to [\pm 1]^d$ with $\ex{M'}{\|M'(D') - \overline{D'}\|_p} \leq d^{1/p}/10$, then we obtain a lower bound of $n \geq \Omega(n^*/\alpha\varepsilon)$ for all $(\varepsilon,\varepsilon\delta)$-differentially private mechanisms $M \from \pmo^{n \times d} \to [\pm 1]^d$ with $\ex{M}{\|M(D) - \overline{D}\|_p} \leq \alpha d^{1/p}$. So we will simply fix the parameters $\alpha=1/10$ and $\varepsilon=1$ in our lower bounds.

\jnote{All the theorems are very explicit about what we mean by one-way marginals, and its stated formally in the intro.  Is it worth adding something for people who want to jump right into the technical parts?}

\section{Lower Bounds for Approximate Differential Privacy}
Our main theorem can be stated as follows.
\begin{thm}[Main Theorem] \label{thm:Main}
Let $M : \pmo^{n \times d} \to [\pm 1]^d$ be a $(1, \delta)$-differentially private mechanism that answers one-way marginals such that $$\forall D \in \pmo^{n \times d}~~~\ex{M}{\norm{M(D)-\overline{D}}_1} \leq \frac{d}{10},$$ where $\overline{D}$ is the true answer vector. If $2^{-\Omega(n)} \leq \delta \leq 1/n^{1+\Omega(1)}$ and $n$ is sufficiently large, then $$d \leq O\left(\frac{n^2}{\log(1/\delta)} \right).$$
\end{thm}

Theorem~\ref{thm:maininformal} in the introduction follows by 
rearranging terms, and applying Fact~\ref{fact:alphaeps}.  The statement above is more convenient technically, but the statement in the introduction is more consistent with the literature.

First we must introduce fingerprinting codes. The following definition is tailored to the application to privacy.  Fingerprinting codes were originally defined by Boneh and Shaw \cite{BonehS98} with a worst-case accuracy guarantee. Subsequent  works \cite{BunUV14, SteinkeU14} have altered the accuracy guarantee to an average-case one, which we use here.

\begin{definition}[$L_1$ Fingerprinting Code] \label{def:FPC}
A \emph{$\varepsilon$-complete $\delta$-sound $\alpha$-robust $L_1$ fingerprinting code for $n$ users with length $d$} is a pair of random variables $D \in \{\pm 1\}^{n \times d}$ and $\trace : [\pm 1]^d \to 2^{[n]}$ such that the following hold.
\begin{itemize}
\item[] \textbf{Completeness:} For any fixed $M : \{\pm 1\}^{n \times d} \to [\pm 1]^d$, $$\pr{}{\left(\norm{M(D)-\overline{D}}_1\leq \alpha d\right) \land \left(\trace(M(D)) = \emptyset\right) } \leq \varepsilon.$$
\jnote{We should be consistent about removing rows versus replacing rows.}
\item[] \textbf{Soundness:} For any $i \in [n]$ and fixed $M : \{\pm 1\}^{n \times d} \to [\pm 1]^d$, $$\pr{}{i \in \trace(M(D_{-i})) } \leq \delta,$$ where $D_{-i}$ denotes $D$ with the $i^\text{th}$ row replaced by some fixed element of $\pmo^d$.
\end{itemize}
\end{definition}

Fingerprinting codes with optimal length were first constructed by Tardos \cite{Tardos03} (for worst-case error) and subsequent works \cite{BunUV14,SteinkeU14} have adapted Tardos' construction to work for average-case error guarantees, which yields the following theorem. 

\jnote{We should be clear whether there is anything new here or if we're just including this for completeness (in which case we have to cite something).}
\begin{theorem} \label{thm:FPC}
For every $n \geq 1$, $\delta > 0$, and $d \geq d_{n,\delta} = O(n^2 \log(1/\delta))$, there exists a $1/100$-complete $\delta$-sound $1/8$-robust $L_1$ fingerprinting code for $n$ users with length $d$.
\end{theorem}

We now show how the existence of fingerprinting codes implies our lower bound.

\begin{proof}[Proof of Theorem \ref{thm:Main} from Theorem \ref{thm:FPC}]
Let $M : \pmo^{n \times d} \to [\pm 1]^d$ be a $(1,\delta)$-differentially private mechanism such that $$\forall D \in \pmo^{n \times d}~~~\ex{M}{\norm{M(D)-\overline{D}}_1} \leq \frac{d}{10}.$$ Then, by Markov's inequality, \begin{equation} \label{eqn:Markov} \forall D \in \pmo^{n \times d}~~~\pr{M}{\norm{M(D)-\overline{D}}_1 > \frac{d}{9}} \leq \frac{9}{10}. \end{equation}

Let $k$ be a parameter to be chosen later. Let $n_k = \lfloor n/k \rfloor$. Let $M_k : \pmo^{n_k \times d} \to [\pm 1]^d$ be the following mechanism. On input $D^* \in \pmo^{n_k \times d}$, $M_k$ creates $D \in \pmo^{n \times d}$ by taking $k$ copies of $D^*$ and filling the remaining entries with 1s. Then $M_k$ runs $M$ on $D$ and outputs $M(D)$.

By group privacy (Fact \ref{fact:GroupDP}), $M_k$ is a $\left(\eps_k = k, \delta_k =  \frac{e^{k}-1}{e-1}\delta\right)$-differentially private mechanism. By the triangle inequality, \begin{equation} \label{eqn:Triangle} \norm{M_k(D^*)-\overline{D^*}}_1 \leq \norm{M(D)-\overline{D}}_1 + \norm{\overline{D}-\overline{D^*}}_1.\end{equation}
Now $$\overline{D}_j = \frac{k \cdot n_k}{n} \overline{D^*_j} + \frac{n-k \cdot n_k}{n} 1.$$
Thus $$\left|\overline{D}_j - \overline{D^*_j} \right| = \left| \left(\frac{k \cdot n_k }{n}-1\right) \overline{D^*_j} + \frac{n-k \cdot n_k}{n} \right| = \frac{n-k \cdot n_k}{n} \left|1-\overline{D^*_j}\right| \leq 2 \frac{n-k \cdot n_k}{n}.$$
We have $$\frac{n-k \cdot n_k}{n} = \frac{n-k\lfloor n/k \rfloor}{n} \leq \frac{n-k(n/k-1)}{n} = \frac{k}{n}.$$ Thus $\norm{\overline{D}-\overline{D^*}}_1 \leq 2k/n$. Assume $k \leq n/200$. Thus $\norm{\overline{D}-\overline{D^*}}_1 \leq d/100$ and, by \eqref{eqn:Markov} and \eqref{eqn:Triangle}, \begin{equation}\label{eqn:MkL1} \pr{M_k}{\norm{M_k(D^*)-\overline{D^*}}_1 > \frac{d}{8}} \leq \pr{M}{\norm{M(D)-\overline{D}}_1 > \frac{d}{9}} \leq \frac{9}{10}.\end{equation}

Assume $d \geq d_{n_k,\delta}$, were $d_{n_k,\delta}=O(n_k^2 \log(1/\delta))$ is as in Theorem \ref{thm:FPC}. We will show by contradiction that this cannot be -- that is $d \leq O(n_k^2 \log(1/\delta))$. Let $D^* \in \{\pm 1\}^{n_k \times d}$ and $\trace : [\pm 1]^d \to 2^{[n_k]}$ be a $1/100$-complete $\delta$-sound $1/8$-robust $L_1$ fingerprinting code for $n_k$ users of length $d$.

By the completeness of the fingerprinting code, \begin{equation}\label{eqn:Complete100} \pr{}{\norm{M_k(D^*)-\overline{D^*}}_1 \leq \frac{d}{8} \wedge \trace(M(D)) = \emptyset } \leq \frac{1}{100}.\end{equation} Combinging \eqref{eqn:MkL1} and \eqref{eqn:Complete100}, gives $$\pr{}{\trace(M_k(D^*)) \ne \emptyset } \geq \frac{9}{100} >\frac{1}{12}.$$ In particular, there exists $i^* \in [n_k]$ such that \begin{equation} \label{eqn:Complete} \pr{}{i^* \in \trace(M_k(D^*))}  > \frac{1}{12n_k}.\end{equation}

We have that $\trace(M_k(D^*))$ is a $(\varepsilon_k, \delta_k)$-differentially private function of $D^*$, as it is only postprocessing $M_k(D^*)$. Thus \begin{equation} \label{eqn:PrivSound} \pr{}{i^* \in \trace(M_k(D^*))} \leq e^{\varepsilon_k} \pr{}{i^* \in \trace(M_k(D^*_{-i^*}))} + \delta_k \leq e^{\varepsilon_k} \delta + \delta_k,\end{equation}
where the second inequality follows from the soundness of the fingerprinting code.

Combining \eqref{eqn:Complete} and \eqref{eqn:PrivSound} gives \begin{equation} \label{eqn:Combine} \frac{1}{12n_k} \leq e^{\varepsilon_k} \delta + \delta_k = e^{k} \delta + \frac{e^{k}-1}{e-1}\delta = \frac{e^{k+1}-1}{e-1}\delta < e^{k+1} \delta.\end{equation} If $k \leq \log(1/12n_k\delta) -1$, then \eqref{eqn:Combine} gives a contradiction. Let $k = \lfloor \log(1/12n\delta) -1\rfloor$. Assuming $\delta \geq e^{-n/200}$ ensures $k \leq n/200$, as required. Assuming $\delta \leq 1/n^{1+\gamma}$ implies $k\geq \log(1/\delta)/(1+1/\gamma)-5 \geq \Omega(\log(1/\delta))$. This setting of $k$ gives a contradiction, which implies that $$d < d_{n_k,\delta} = O(n_k^2 \log(1/\delta)) = O\left(\frac{n^2}{k^2} \log(1/\delta)\right) = O\left(\frac{n^2}{\log(1/\delta)}\right),$$ as required.

\end{proof}

\section{New Mechanisms for $L_{\infty}$ Error}

Adding independent noise seems very natural for one-way marginals, but it is suboptimal if one is interested in worst-case (i.e. $L_\infty$) error bounds, rather than average-case (i.e. $L_1$) error bounds.

\subsection{Pure Differential Privacy}

Theorem \ref{thm-intro:LinfPure} follows from Theorem \ref{thm:LinfPure}. In particular, the mechanism $M : \pmo^{n \times d} \to [\pm 1]^d$ in Theorem \ref{thm-intro:LinfPure} is given by $M(D) = \overline{D}+Y$, where $Y \sim \mathcal{D}$ and $\mathcal{D}$ is the distribution from Theorem \ref{thm:LinfPure} with $\Delta=2/n$.\footnote{Note that we must truncate the output of $M$ to ensure that $M(D)$ is always in $[\pm 1]^d$.}

\begin{thm} \label{thm:LinfPure}
For all $\varepsilon > 0$, $d \geq 1$, and $\Delta > 0$, there exists a continuous distribution $\mathcal{D}$ on $\mathbb{R}^d$ with the following properties.
\begin{itemize}
\item \textbf{Privacy:} If $x,x' \in \mathbb{R}^d$ with $\norm{x-x'}_\infty \leq \Delta$, then $$\pr{Y \sim \mathcal{D}}{x+Y \in S} \leq e^{\varepsilon} \pr{Y \sim \mathcal{D}}{x'+Y \in S}$$ for all measurable $S \subseteq \mathbb{R}^d$.
\item \textbf{Accuracy:} For all $\alpha > 0$, $$\pr{Y \sim \mathcal{D}}{\norm{Y}_\infty \geq \alpha} \leq \left(\frac{\Delta d}{\varepsilon \alpha} \right)^d e^{d-\alpha\varepsilon/\Delta}.$$
In particular, if $d \leq \varepsilon\alpha/2\Delta$, then $\pr{Y \sim \mathcal{D}}{\norm{Y}_\infty \geq \alpha} \leq (2e)^{-d}$.
\item \textbf{Efficiency:} $\mathcal{D}$ can be efficiently sampled.
\end{itemize}
\end{thm}

\begin{proof}
The distribution $\mathcal{D}$ is simply an instantiation of the exponential mechanism \cite{McSherryT07}. In particular, the probability density function is given by $$\text{pdf}_\mathcal{D}(y) \propto \exp\left(-\frac{\varepsilon}{\Delta}\norm{y}_\infty\right).$$
Formally, for every measurable $S \subseteq \mathbb{R}^d$, $$\pr{Y \sim \mathcal{D}}{Y \in S} = \frac{\int_{S} \exp\left(-\frac{\varepsilon}{\Delta}\norm{y}_\infty\right) \mathrm{d}y}{\int_{\mathbb{R}^d} \exp\left(-\frac{\varepsilon}{\Delta}\norm{y}_\infty\right) \mathrm{d}y}.$$
Firstly, this is clearly a well-defined distribution as long as $\varepsilon/\Delta>0$.

Privacy is easy to verify: It suffices to bound the ratio of the probability densities for the shifted distributions. For $x,x' \in \mathbb{R}^d$ with $\norm{x'-x}_\infty \leq \Delta$, by the triangle inequality, $$\frac{\text{pdf}_\mathcal{D}(x+y)}{\text{pdf}_\mathcal{D}(x'+y)} = \frac{\exp\left(-\frac{\varepsilon}{\Delta}\norm{x+y}_\infty\right)}{\exp\left(-\frac{\varepsilon}{\Delta}\norm{x'+y}_\infty\right)} = \exp\left(\frac{\varepsilon}{\Delta} \left( \norm{x'+y}_\infty - \norm{x+y}_\infty \right) \right) \leq \exp\left(\frac{\varepsilon}{\Delta} \norm{x'-x}_\infty  \right) \leq e^\varepsilon.$$ 

Define a distribution $\mathcal{D}^*$ on $[0,\infty)$ to by $Z \sim \mathcal{D}^*$ meaning $Z=\norm{Y}_\infty$ for $Y \sim \mathcal{D}$. To prove accuracy, we must give a tail bound on $\mathcal{D}^*$. The probability density function of $\mathcal{D}^*$ is given by $$\text{pdf}_{\mathcal{D}^*}(z) \propto z^{d-1} \cdot \exp\left(-\frac{\varepsilon}{\Delta}z\right),$$ which is obtained by integrating the probability density function of $\mathcal{D}$ over the infinity-ball of radius $z$, which has surface area $d2^dz^{d-1} \propto z^{d-1}$. Thus $\mathcal{D}^*$ is precisely the gamma distribution with shape $d$ and mean $d\Delta/\varepsilon$. 
The moment generating function is therefore $$\ex{Z \sim \mathcal{D}^*}{e^{tZ}} = \left(1-\frac{\Delta}{\varepsilon} t\right)^{-d}$$ for all $t < \varepsilon/\Delta$. By Markov's inequality $$\pr{Z \sim \mathcal{D}^*}{Z \geq \alpha} \leq \frac{\ex{Z \sim \mathcal{D}^*}{e^{tZ}}}{e^{t\alpha}} =  \left(1-\frac{\Delta}{\varepsilon} t\right)^{-d} e^{-t\alpha}.$$ Setting $t=\varepsilon/\Delta - d/\alpha$ gives the required bound.

It is easy to verify that $Y \sim \mathcal{D}$ can be sampled by first sampling a radius $R$ from a gamma distribution with shape $d+1$ and mean $(d+1)\Delta/\varepsilon$ and then sampling $Y \in [\pm R]^d$ uniformly at random. To sample $R$ we can set $R = \frac{\Delta}{\varepsilon}\sum_{i=0}^d \log U_i$, where each $U_i \in (0,1]$ is uniform and independent. This gives an algorithm (in the form of an explicit circuit) to sample $\mathcal{D}$ that uses only $O(d)$ real arithmetic operations, $d+1$ logarithms, and $2d+1$ independent uniform samples from $[0,1]$. 

\end{proof}

\subsection{Approximate Differential Privacy}

\jnote{Still extremely rough.  Applications of sparse vector are such a bitch to formalize.}

Our algorithm for approximate differential privacy makes use of a powerful tool from the literature \cite{DworkNRRV09,HardtR10,DworkNPR10,RothR10} called the sparse vector algorithm:

\begin{thm}[Sparse Vector] \label{thm:SV}
For every $c,k \geq 1$, $\varepsilon,\delta,\alpha,\beta>0$, and $$n \geq O \left(\frac{\sqrt{c\log(1/\delta)} \log(k/\beta)}{\alpha\varepsilon} \right),$$ there exists a mechanism $\mathit{SV}$ with the following properties.
\begin{itemize}
\item $\mathit{SV}$ takes as input a database $D \in \mathcal{X}^n$ and provides answers $a_1, \cdots, a_k \in [\pm 1]$ to $k$ (\emph{adaptive}) linear queries $q_1, \cdots, q_k : \mathcal{X} \to [\pm 1]$.
\item $\mathit{SV}$ is $(\varepsilon,\delta)$-differentially private.
\item Assuming $$\left|\left\{ j \in [k] : |q_j(D)| > \alpha/2 \right\}\right| \leq c,$$ we have $$\pr{\mathit{SV}}{\forall j \in [k] ~~ |a_j-q_j(D)| \leq \alpha} \geq 1-\beta.$$
\end{itemize}
\end{thm}

A proof of this theorem can be found in \cite[Theorem 3.28]{DworkR13}.\footnote{Note that the algorithms in the literature are designed to sometimes output $\perp$ as an answer or halt prematurely.  To modify these algorithms into the form given by Theorem \ref{thm:SV} simply  output $0$ in these cases.}
We now describe our approximately differentially private mechanism.

\begin{figure}[ht] 
\begin{framed}
\begin{algorithmic}
\INDSTATE[0]{}
\INDSTATE[0]{Parameters: $\varepsilon,\delta>0$.}
\INDSTATE[0]{Input: $D \in \{\pm 1\}^{n \times d}$.}
\INDSTATE[0]{Let $$\sigma = 5\sqrt{d\log(1/\delta)}/\eps n \qquad \qquad \alpha = 8\sigma \sqrt{\log\log d}$$}
\INDSTATE[0]{For $j \in [d]$, let $\tilde{a}_j = \overline{D}_j + z_j$ where $z_j \sim N(0, \sigma^2)$.}
\INDSTATE[0]{Instantiate $\mathit{SV}$ from Theorem \ref{thm:SV} with parameters 
$$
c_\mathit{SV}=2d/\log^8 d \qquad k_\mathit{SV}=d \qquad \varepsilon_\mathit{SV} = \varepsilon/2 \qquad \delta_\mathit{SV} = \delta/2
$$
$$
\alpha_\mathit{SV} = \alpha/2 \qquad\beta_\mathit{SV}=e^{-\log^4 d}
$$}
\INDSTATE[0]{For $j \in [d]$, define $q_j : \{\pm 1\}^d \to [\pm 1]$ by $q_j(x) = (x_j-\tilde{a}_j)/2$.}
\INDSTATE[0]{Let $\hat{a}_1, \cdots, \hat{a}_d$ be the answers to $q_1, \cdots, q_d$ given by $\mathit{SV}$.}
\INDSTATE[0]{For $j \in [d]$, let $a_j = \tilde{a}_j + 2\hat{a}_j$.}
\INDSTATE[0]{Output $a_1, \cdots, a_d$.}
\end{algorithmic}
\end{framed}
\vspace{-6mm}
\caption{Approximately DP mechanism $M: \{\pm 1\}^{n \times d} \to [\pm 1]^d$}
\end{figure}

\begin{proof}[Proof of Theorem \ref{thm-intro:LinfApprox}.]


Firstly, we consider the privacy of $M$: $\tilde{a}$ is the output of the Gaussian mechanism with parameters to ensure that it is a $(\varepsilon/2,\delta/2)$-differentially private function of $D$. Likewise $\hat{a}$ is the output of $\mathit{SV}$ with parameters to ensure that it is also a $(\varepsilon/2,\delta/2)$-differentially private function of $D$. Since the output is $\tilde{a}+2\hat{a}$, composition implies that $\mathcal{M}$ as a whole is $(\varepsilon,\delta)$-differentially private, as required.

Now we must prove accuracy. Suppose that $|\hat{a}_j - q_j(D)| \leq \alpha_\mathit{SV} = \alpha/2$ for all $j \in [d]$. Then
\begin{align*}
 |a_j - \overline{D}_j| =& |\tilde{a}_j + 2\hat{a}_j - \overline{D}_j|\\
=& |\tilde{a}_j -\overline{D}_j+ 2(q_j(D) + (\hat{a}_j - q_j(D)))|\\
\leq& |\tilde{a}_j -\overline{D}_j+ 2q_j(D)| + 2|\hat{a}_j - q_j(D))|\\
\leq& |\tilde{a}_j -\overline{D}_j+ (\overline{D}-\tilde{a}_j)| + 2\alpha_\mathit{SV}\\
=& \alpha,
\end{align*}
as required. So we need only show that $|\hat{a}_j - q_j(D)| \leq \alpha_\mathit{SV} $ for all $j \in [d]$, which sparse vector guarantees will happen with probability at least $1-\beta_\mathit{SV}$ as long as \begin{equation} \label{eqn:Sparse} \left|\left\{ j \in [d] : |q_j(D)| > \alpha_\mathit{SV}/2 \right\}\right| \leq c_\mathit{SV}.\end{equation} Now we verify that \eqref{eqn:Sparse} holds with high probability.

By our setting of parameters, we have $q_j(D) = -z_j/2$. This means $$\pr{}{|q_j(D)| > \alpha_\mathit{SV}/2} = \pr{}{|z_j| > \alpha/2} \leq e^{-\alpha^2/8\sigma^2} = \frac{1}{\log^8 d}.$$
Let $E_j \in \{0,1\}$ be the indicator of the event $|q_j(D)| > \alpha_\mathit{SV}/2$. Since the $z_j$s are independent, so are the $E_j$s. Thus we can apply a Chernoff bound: \begin{equation} \label{eqn:SparseWHP} \pr{}{\left|\left\{ j \in [d] : |q_j(D)| > \alpha_\mathit{SV}/2 \right\}\right| > c_\mathit{SV}} = \pr{}{\sum_{j \in [d]} E_j > \frac{2d}{\log^8 d}} \leq e^{-2d/\log^{16} d}.\end{equation}
The failure probability of $M$ is bounded by the failure probability of $\mathit{SV}$ plus \eqref{eqn:SparseWHP}, which is dominated by $\beta_\mathit{SV} = \exp(-\log^4 d)$.

Finally we consider the sample complexity. The accuracy is bounded by 
$$\alpha \leq \frac{40 \sqrt{d \cdot \log(1/\delta) \cdot \log \log d}}{\varepsilon n},$$ which rearranges to $$n \geq \frac{40 \sqrt{d \cdot \log(1/\delta) \cdot \log \log d} }{\alpha \varepsilon}.$$ Theorem \ref{thm:SV} requires $$n \geq O\left(\frac{\sqrt{c_\mathit{SV} \log(1/\delta)} \log(d/\beta_\mathit{SV})}{\alpha\varepsilon}\right) = O\left(\frac{\sqrt{d \log(1/\delta)}}{\alpha\varepsilon}\right)$$ for sparse vector to work, which is also satisfied.
\end{proof}
We remark that we have not attempted to optimize the constant factors in this analysis.

\addcontentsline{toc}{section}{References}
\bibliographystyle{alpha}
\bibliography{references}

\appendix

\section{Alternative Lower Bound for Pure Differential Privacy}

It is known \cite{HardtT10} that any $\varepsilon$-differentially private mechanism that answers $d$ one-way marginals requires $n \geq \Omega(d/\varepsilon)$ samples. Our techniques yield an alternative simple proof of this fact.

\begin{thm} \label{thm:Packing}
Let $M : \{\pm 1\}^{n \times d} \to [\pm 1]^d$ be a $\varepsilon$-differentially private mechanism. Suppose $$\forall D \in \{\pm 1\}^{n \times d}~~~\ex{M}{\norm{M(D)-\overline{D}}_1} \leq 0.9d$$ Then $n \geq \Omega(d/\varepsilon)$.
\end{thm}

The proof uses a special case of Hoeffding's Inequality:
\begin{lem}[Hoeffding's Inequality] \label{lem:HoeffdingInequality}
Let $X \in \{\pm 1\}^n$ be uniformly random and $a \in \mathbb{R}^n$ fixed. Then $$\pr{X}{ \langle a, X \rangle > \lambda \norm{a}_2} \leq e^{-\lambda^2/2}$$ for all $\lambda \geq 0$.
\end{lem}

\begin{proof}[Proof of Theorem \ref{thm:Packing}]
Let $x,x' \in \{\pm 1\}^d$ be independent and uniform. Let $D \in \{\pm 1\}^{n \times d}$ be $n$ copies of $x$ and, likewise, let $D' \in \{\pm 1\}^{n \times d}$ be $n$ copies of $x'$.
Let $Z = \langle M(D), x \rangle$ and $Z' = \langle M(D'), x \rangle$.

Now we give conflicting tail bounds for $Z$ and $Z'$, which we can relate by privacy.

By our hypothesis and Markov's inequality, 
\begin{align*}
\pr{}{Z \leq d/20} =& \pr{}{\langle M(D), x \rangle \leq 0.05d}\\
=& \pr{}{\langle \overline{D}, x \rangle - \langle \overline{D}-M(D), x \rangle \leq 0.05d}\\
=& \pr{}{\langle \overline{D}-M(D), x \rangle \geq 0.95d}\\
\leq& \pr{}{\norm{\overline{D}-M(D)}_1 \geq 0.95d}\\
\leq& \frac{\ex{}{\norm{\overline{D}-M(D)}_1}}{0.95d}\\
\leq& \frac{0.9}{0.95} < 0.95.
\end{align*}

Since $M(D')$ is independent from $x$, we have $$\forall \lambda \geq 0 ~~~\pr{}{Z' > \lambda \sqrt{d}} \leq \pr{}{\langle M(D'), x \rangle > \lambda \norm{M(D')}_2} \leq e^{-\lambda^2/2},$$ by Lemma \ref{lem:HoeffdingInequality}. In particular, setting $\lambda = \sqrt{d}/20$ gives $\pr{}{Z' > d/20} \leq e^{-d/800}$.

Now $D$ and $D'$ are databases that differ in $n$ rows, so privacy implies that $$\pr{}{M(D) \in S} \leq e^{n\varepsilon} \pr{}{M(D') \in S}$$ for all $S$. Thus $$\frac{1}{20} < \pr{}{Z > \frac{d}{20}} = \pr{}{M(D) \in S_x} \leq e^{n\varepsilon} \pr{}{M(D') \in S_x} = e^{n\varepsilon} \pr{}{Z' > \frac{d}{20}} \leq e^{n\varepsilon} e^{-d/800},$$ where $$S_x = \left\{ y \in [\pm 1]^d : \langle y, x \rangle > \frac{d}{20} \right\}.$$ Rearranging $1/20 < e^{n \varepsilon} e^{-d/800}$, gives $$n > \frac{d}{800\varepsilon} - \frac{\log(20)}{\varepsilon},$$
as required.
\end{proof}

\end{document}